\documentclass[11pt]{article}
\usepackage[utf8]{inputenc}

\usepackage[margin=1in,papersize={8.5in,11in}]{geometry}
\usepackage{times,xcolor,hyperref}
\usepackage{amssymb,amsfonts,amsmath,mathrsfs,mathtools,amsthm}
\usepackage{physics}
\usepackage{graphicx,epsfig,subfigure}
\usepackage{bm}

\newcommand\Min{\mathrm{min}}
\newcommand\Max{\mathrm{max}}

\DeclareMathAlphabet{\mathpzc}{OT1}{pzc}{m}{it}

\newtheorem{theorem}{Theorem}

\newtheorem{lemma}{Lemma}
\newtheorem{proposition}[theorem]{Proposition}
\newcommand{\ca}{\mathcal{A}}
\newcommand{\ci}{\mathcal{I}}

\newcommand{\h}{\mathpzc{h}}
\newcommand{\D}{\mathpzc{D}}
\newcommand{\CO}[1]{\mathcal{O}(#1) }
\newcommand{\comment}[1]{}

\title{Enabling Quantum Speedup of Markov Chains using a Multi-level Approach}
\author{Xiantao Li\\
The Pennsylvania State University, University Park, Pennsylvania, 16802, USA\\
Xiantao.Li@psu.edu
}

\begin{document}
\maketitle
\begin{abstract}
Quantum speedup for mixing a Markov chain can be achieved based on the construction of a slowly-varying $r$ Markov chains where  the initial chain can be easily prepared and the spectral gaps have uniform lower bound. The overall complexity is proportional to $r$. We present a multi-level approach to construct such sequence of $r$ Markov chains by varying a  resolution parameter $\h.$  We show that  the density function of a low-resolution Markov chain can be used to warm start the Markov chain with high resolution.  We prove that in terms of the chain length the new algorithm has $O(1)$ complexity  rather than $O(r).$
\end{abstract}

\section{Introduction}

Markov chains play a central role in modern data science algorithms \cite{ching2013markov}, e.g., data assimilation and forecasting, uncertainty quantification, machine learning, stochastic optimization etc. One important scenario is when the mixing properties of Markov chains are used to sample statistical quantifies of interest with respect to the equilibrium density $\pi$. However, simulating a  Markov chain for this purpose can be a computationally demanding task, due to the dimensionality and mixing time. The seminal works Aharonov et al. and Szegedy \cite{aharonov2001quantum,szegedy2004quantum} put forward a paradigm to simulate Markov chains using quantum walks. 
A particular emphasis in this line of works \cite{wocjan2008speedup,richter2007almost,magniez2007search,dunjko2015quantum, wocjan2021szegedy} has been placed on the speedup in terms of the spectral gap $\delta$, from the classical $\frac{1}{\delta}$ complexity to $\frac{1}{\sqrt{\delta}}$, without the unfavorable dependence on $\pi_*=\min_{x}\pi(x)$.  
An important approach for achieving this quadratic speedup is by constructing a sequence of Markov chains, $P_0, P_1, \cdots, P_r$, for which the stationary distribution of two successive Markov chains are close  \cite{aharonov2003adiabatic,wocjan2008speedup,somma2008quantum}. Assuming that the initial chain $P_0$ can be easily mixed, the overall complexity is $\mathcal{O}\left(\frac{r}{\sqrt{\delta}}\right),$ times the complexity of implementing each quantum walk.
This approach has been combined with the Chebyshev cooling schedule, together with the quantum mean estimation to compute expected values and partition functions \cite{montanaro2015quantum}, which yields the sampling complexity $\mathcal{O}\left(\frac{r}{\epsilon\;\sqrt{\delta}}\right)$ . 
One important example of slowly varying Markov chains is the Gibbs distribution parameterized by the inverse temperature, which is the primary mechanism behind simulated annealing algorithms. Although the quadratic speedup has been envisioned to be a generic property, explicit construction of the sequence of slowly varying Markov chains with uniformly bounded spectral gaps is still an open issue.

This paper presents an alternative approach to construct multiple Markov chains. We consider the discretization of a Markov chain with continuous (infinite-dimensional) state space, which stem from general deterministic or stochastic dynamical systems. The discretization, using Ulam-Galerkin projection \cite{ulam1960collection}, lumps states into finitely many bins with bin size $\h$. The novel aspect is that we can construct multiple Markov chains $P_\h$ by varying $\h$. When $\h$ is large, e.g., $\h=\h_\Max$, the state space of the Markov chain is small, $P_\h$ can be mixed quickly using either a classical  or a quantum algorithm. On the other hand, when $\h$ is small, e.g., $\h=\h_\Min$, the continuous Markov chain is well approximated by $P_\h$. To some extent, this removes the assumption in the framework of multiple Markov chains \cite{wocjan2008speedup} on the preparability of the initial chain.  We show that by varying $\h$, e.g., from $2\h$ to $\h$ at each stage, the density functions of the Markov chains  at two successive levels have significant overlap, thus enabling a smooth transition. 
Our main finding (Theorem 1) is that simulating the sequence of such multiple Markov chains $\{P_\h\}_{\h_\Max \geq \h \geq h_\Min}$ has a cost that is comparable to simulating the  Markov chain $P_{\h_\Min}$, as if it had a warm start. 

Although our approach is constructed from a finite-dimensional approximation of a Markov chain with infinite state space, the same methodology can be applied directly to certain finite Markov chains, especially those that have been treated by  multigrid methods \cite{de2010smoothed}. 
\medskip

\emph{Problem Setup. ---}  We consider a Markov chain $\{X_n\}_{n\geq 0}$ with state space $S=\mathbb{R}^d$ and the (right) transition density $K(x,y)$, with $K(x,dy)$ indicating the probability that, given $x$, the Markov chain moves to state $y$ at the next step, which can be described in terms of  the conditional expectation of any statistical quantity $f(\cdot)$, i.e., 
$\mathbb{E}[f(X_{n+1})|X_n]=\int K(X_n,y) f(y) dy.$ An alternative description of the Markov chain is the Chapman-Kolmogorov (CK) equation for the change of the probability density function (PDF) from step $n$ to $n+1,$
\begin{equation} \label{eq: chkol}
    p_{n+1} = \int_{\mathbb{R}^d} p_{n}(x) K(x, y)  dx, 
\end{equation}
which is convenient in  the study of mixing properties.  The relation in Eq. \eqref{eq: chkol} is often written in a matrix/vector multiplication form $p_{n+1}^T= p_n^T K$. In particular, a stationary density $p(x)$ is the left eigenvector associated with the eigenvalue $1$: $p^T= p^T K$. 

\medskip 
\noindent\emph{Problem: } Given a precision $\epsilon,$ find a finite-dimensional approximation $p_\h(x)$ of the stationary PDF $p$ with $\h$ indicating a numerical parameter, such that, $\norm{p - p_\h}_1 < \epsilon. $

\medskip 
\emph{Ulam-Galerkin projection. --- } Many existing quantum algorithms work with a finite-dimensional form of the CK equation \eqref{eq: chkol}. To approximate a Markov chain with continuous state space and implement the algorithms on quantum computers, we have to quantize the problem. This is done in two steps: First we consider a large domain $\D$, where the probability of reaching states outside $\D$ is negligible. This is a reasonable assumption if we consider the states close to equilibrium: since the PDF integrates to 1, the probability in the far field is negligible.  For simplicity of the presentation, we choose $\D= [-1,1]^{\otimes d}$. In the second step, we can introduce a partition of $\D$ with uniform spacing  $\h$, $(\h\coloneqq \frac{1}{N}$ for some $N \in \mathbb{N}$),
\begin{equation}
    \D = \bigcup_{\bm j\in N(\h) }  \D_{\bm j }(\h), N(\h):= \Big\{\bm j: -N\leq j_k < N, \forall  k\in [d]\Big\}.
\end{equation}
where the subdomain is given by,
\begin{equation}
   \D_{\bm j}(\h)=  \bigotimes_{k=1}^d \big[j_k  \h, (j_k+1)  \h\big).
\end{equation}

With this partition, the continuous states have been grouped into small non-intersecting bins with bin size $\h^d$, amounting to a finite state space,
\begin{equation}\label{Sh}
    S_\h \coloneqq \{  \D_{\bm j }(\h)| \bm j\in N(\h) \}, \quad \abs{S_\h}= \left(\frac{2}{\h}\right)^d.
\end{equation}
For brevity, we simplify refer to one such state in $S_\h$ by $\bm j.$ 
Associated with the partitions of the domain is a discretization of a PDF.  This can be done by  integrating the PDF $p(x)$ over each bin:
\begin{equation}\label{lump}
    \pi_\h(\bm j)= \int_{\Omega_{\bm j }(\h)} p(x) dx.
\end{equation}  

Meanwhile, the discrete probability can be mapped back to a continuous one by piecewise constant interpolation, 
\begin{equation}\label{interp}
    p_\h(x) = \sum_{\bm j} \pi_\h ({\bm j}) \chi_{\bm j}(x). 
\end{equation}
Here $\chi_{\bm j}(x)$ is the indicator function for the domain $\Omega_{\bm j }(\h)$: $\chi_{\bm j}(x)=\h^{-d},$ if $x\in \Omega_{\bm j }(\h)$, and zero otherwise. We denote such class of functions by $\Delta_\h.$

\emph{Approximation properties. --- } Clearly the function $p_\h$ defined in \eqref{interp} is also a PDF.  In addition, \eqref{interp} indicates a one-to-one correspondence between the vector $\bm \pi_\h\in \mathbb{R}^{\abs{S_\h}}$ and a piecewise constant function $p_\h$ in $\Delta_\h.$ We will write the relation in \eqref{lump} as $\bm \pi_\h= \ca^\h p(x)$ with $\ca^\h$ representing an averaging operator, and  \eqref{interp} as an interpolation $p_\h= \ci_\h  \bm \pi_\h$.
The $L_1$ norm of $p_\h$ also coincides with the vector $1$-norm of $\pi_\h,$ so we simply 
denote this norm by $\norm{\cdot}_1.$ We first provide an error bound for an approximation of $p(x)$ using piecewise constant functions  \cite{devore1993constructive}, i.e., those in $\Delta_\h$.

\begin{lemma}\label{lmm: errbd}
Assume that $p(x)$ is Lipschitz continuous with Lipschitz constant $\Lambda.$ 
\[\abs{ p(x) - p(y) } \leq \Lambda \norm{x-y}_{\infty}.\]
When $p(x)$ is approximated by \eqref{interp} with coefficients from \eqref{lump}, 
the following estimate holds, 
\begin{equation}
   \norm{ p(x) -  p_\h(x)}_1 \leq \Lambda \h, \quad  \norm{ p(x) -  p_\h(x)}_\infty \leq \Lambda \h.
\end{equation}
\end{lemma}

 By applying the same procedure to the  transition kernel $K(x,y)$, we obtain a matrix approximation, 
\begin{equation}\label{Ph}
    P_\h(\bm i, \bm j)= \int_{ \Omega_{\bm i}(\h)} \int_{ \Omega_{\bm j}(\h)} K(x,y) dy dx.
\end{equation}
One immediate observation is that $P_\h$ is a finite stochastic matrix, implying that there exists an invariant density $\bm \pi_\h$ (resp. $p_\h$).  Namely, all entries of $P_\h$ are non-negative and the row sum is one.  The stochastic matrix $ P_\h(\bm i, \bm j)$ has a one-to-one correspondence with the continuous transition kernel,
\begin{equation}
    K_\h(x,y) = \sum_{\bm i, \bm j}   P_\h(\bm i, \bm j) \chi_{\bm i}(x)\chi_{\bm j}(y).
\end{equation}

This reduction procedure using piecewise constant interpolation is known as the Ulam-Galerkin projection \cite{ulam1960collection}. Properties of the resulting Markov chains with finite state spaces were part of Ulam's conjectures, some of which have been addressed for several cases \cite{li1976finite,murray1997discrete,ding1996finite,froyland1998approximating}.  Of particular importance to the current context are those results that assert that the density function
$p_\h$ of $K_\h$ converges to the density $p$ of $K$ with rate $\h \log \frac{1}{\h}. $

Quantum algorithms will work with the discrete probability $\{\pi_\h ({\bm j})  \}$. The standard approach is to encode the PDF coherently into a quantum state,
\begin{equation}\label{quan-st}
    \ket{\pi_\h} = \sum_{\bm j} \sqrt{\pi_\h ({\bm j})} \ket{\bm j},
\end{equation}
where the index $\bm j$ is mapped to a computational basis $\ket{\bm j}.$ 
Meanwhile, the transition matrix can be encoded in a unitary operator, either via a Hamiltonian operator \cite{aharonov2003adiabatic,childs2009universal} or quantum walks built from reflections \cite{szegedy2004quantum,wocjan2008speedup}.

By varying the resolution parameter $\h$, we obtain a sequence of Markov chains,
\[\{ P_\mathpzc{h}: \mathpzc{h}_\text{max} \geq  \mathpzc{h} \geq \mathpzc{h}_\text{min} \}. \]
The simplest progression is done by reducing $\h$ by a factor of 2 at each stage. Namely, $\mathpzc{h}_\text{min} = 2^{-r} \mathpzc{h}_\text{max}$. In light of the convergence rate of the Ulam-Galerkin projection, we have 
$   r = \CO{\log \frac{1}{\epsilon}},$
in order for the descritization error to be within $\epsilon.$

\emph{The multilevel approach. --- }The ability to vary $\h$, thus introducing a slow transition from a Markov chain with small state space to one with a larger state space,  suggests a new approach to build a sequence of Markov chains. Specifically, one can start with the equilibrium PDF $\pi_{2\h}$ computed from the Markov chain $P_{2\h}$, followed by an interpolation to prepare the next Markov chain $P_\h$. The overall procedure is  outlined, starting with $ \h=\h_\text{max} $, as follows
{\small 
\begin{equation}\label{diag}
\ket{\pi_\h} \overset{Interp.}{\xrightarrow{\hspace*{1cm}}} \ci_{\h}^{\frac{\h}{2}} \ket{\pi_{{\h}}}  \overset{Q.~ Walks}{\xrightarrow{\hspace*{1cm}}} \ket{\pi_{\frac{\h}{2}}} \overset{Interp.}{\xrightarrow{\hspace*{1cm}}} \ci_{\frac{\h}{2}}^{\frac{\h}{4}} \ket{\pi_{\frac{\h}{2}}} \cdots \overset{Interp.}{\xrightarrow{\hspace*{1cm}}}
\ci_{2\h_\text{min}}^{\h_\text{min}} \ket{\pi_{2\h_\text{min}}}  \overset{Q.~ Walks}{\xrightarrow{\hspace*{1cm}}} \ket{\pi_{\h_\text{min}} }.
\end{equation}
    }

It is worthwhile to point out that our approach departs from the prior works \cite{wocjan2008speedup} in that we work with a sequence of Markov chains with varying state spaces.  When $\mathpzc{h}$ is large, the dimension of the state space is much smaller, which can be treated much more efficiently, thus  relaxing the preparability assumption on the initial chains. In addition, although the lumped Markov chain $P_\h$ shows resemblance with the Markov chain on a lattice considered in \cite{richter2007almost}, $P_\h$ in the current case is not necessarily symmetric and the equilibrium PDF may not be uniform. 

In order to carry out this program, we define unitary operators that map density functions at successive levels. Specifically, following \eqref{lump}, we let  $\ca_\h^{2\h}$ be the ``averaging operator" that reduce PDFs on level $\h$ to PDFs on level $2\h$. Similarly, we define   $\ci_{2\h}^{\h} p_{2\h},$ as the ``interpolation operator", which is similar to \eqref{interp}.

\begin{figure}[hpt]
    \centering
    \includegraphics[scale=0.23]{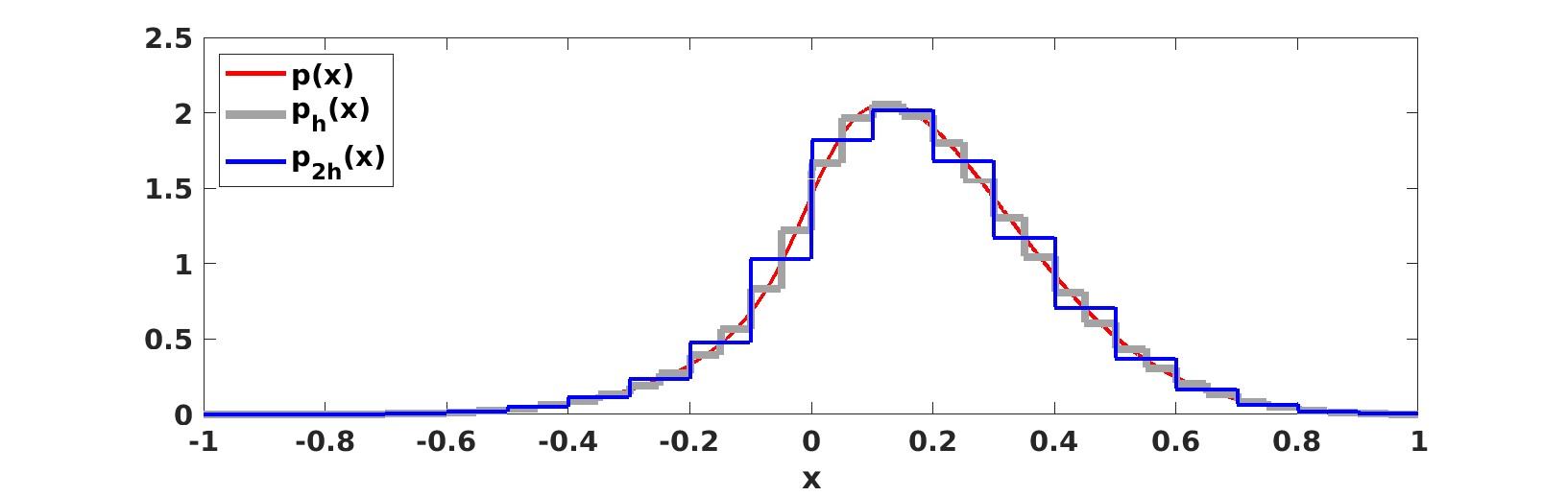}
    \caption{An illustration of the approximate density $p_\h$ and $p_{2\h}$ for a PDF $p(x).$}
    \label{fig:my_label}
\end{figure}

Figure \ref{fig:my_label} provides an illustration of the PDFs at the levels $\h$ and ${2\h}$.  The following properties can be directly established.
\begin{lemma}
  $\ca_\h^{2\h}$ is the operator adjoint of $\ci_{2\h}^\h$ and is proportional to the left inverse of $\ci_{2\h}^\h$: For any $\bm{v}_\h$ (in $ \mathbb{R}^{\abs{S_\h}}$)
 and $\bm v_{2\h}$ (in $ \mathbb{R}^{\abs{S_{2\h}}}$), the following identifies hold
 \begin{equation}\label{aiia}
 \big(\bm v_\h, \ci_{2\h}^{\h} \bm v_{2\h}\big)= \big(\ca_{\h}^{2\h} \bm v_\h, \bm v_{2\h}\big), \quad 
      \ca_{\h}^{2\h} \ci_{2\h}^{\h} \bm v_{2\h}= 2^d \bm v_{2\h},  
  \end{equation}  
where the parenthesis indicate standard inner products.  Furthermore,  the stochastic matrices \eqref{Ph} at consecutive levels are related as follows 
\begin{equation}\label{2h2h}
{P}_{2\h}=\ca_{\h}^{2\h} P_{\h} \ci_{2\h}^\h. 
\end{equation}
\end{lemma}
The last equation \eqref{2h2h} provides an important vehicle to move across Markov chains at different levels. 

\bigskip 
Next, we state the assumptions under which the multilevel procedure has provable performance.

\noindent{\emph{Main Assumptions: Each finite Markov chain $P_\h$ is reversible with stationary PDF, given by $\bm \pi_\h$ (or $p_\h$ according to \eqref{interp}). The spectral gaps   $\delta_\h$ of the Markov chains  are uniformly bounded, in the sense that 
there exists a constant $\gamma$, such that $\frac{\delta_\h}{\delta_{2\h}}\leq  \gamma.$
 Further, each $p_\h(x)$ in the range of $K_\h$ has the following bounded variation property at the coarse level $2\h$,  
\begin{equation}
 \abs{p_\h(x) - p_\h (x')} \leq \Lambda \h, \forall x, x' \in \D_{\bm j}(2\h)
\end{equation}
}}

The uniqueness of the stationary density $p_\h$ has been proved in \cite{de2010smoothed} in the context of multigrid methods. The spectral gaps have a direct influence on the convergence rate of the Markov chain. For the Ulam-Galerkin projection, the uniform bound for the spectral gaps have  been verified for certain stochastic dynamics, e.g., \cite[Cor. 3.5.3]{murray1997discrete}. 

We first show that the uniformity parameter $\gamma$ for the spectral gaps is $\mathcal{O}(1)$, i.e,  $\delta_{2\h} \approx \delta_\h.$. To this end, we use the notion of coefficient $\tau$ of ergodicity by Seneta \cite{seneta1993sensitivity} (which is related to the spectral gap by $\delta= 1 - \tau$),
\begin{equation}\label{eq: tau}
    \tau_\h = \max_{ (\bm \pi_\h, \mathbf{1}_\h)=0 } \frac{ \norm{P_h \bm \pi_\h}_1 }{\norm{\bm \pi_\h}_1}.
\end{equation}
Here $\mathbf{1}_\h$ is the vector with values being all ones.

\begin{lemma}
 Under the assumption above on the bounded variation, one has  $\abs{\tau_\h -\tau_{2\h}} \leq \Lambda \h$.
\end{lemma}

Let $\bm \pi_{2\h}$ be the vector  with $\norm{\bm \pi_{2\h}}_1=1 $ that achieves the maximum for the stochastic matrix $P_{2\h}$, i.e.,
$\tau_{2\h}= \norm{P_{2\h}\bm \pi_{2\h} }_1.$
Then, using the optimality condition \eqref{eq: tau}, we have 
\begin{align*}
    \tau_\h  \ge & \norm{P_h \ci_{2\h}^{\h} \bm \pi_{2\h}}_1 
      = \CO{\h} + \norm{\ci_{2\h}^{\h} \ca_{\h}^{2\h} P_h \ci_{2\h}^{\h} \bm \pi_{2\h}}_1\\
       = & \CO{\h} +\norm{\ci_{2\h}^{\h} P_{2h}  \bm \pi_{2\h}}_1
     =  \CO{\h} + \tau_{2\h} \Longrightarrow \abs{\tau_\h -\tau_{2\h}} \leq \Lambda \h. 
\end{align*}
In the second step, we used the Lemma \ref{lmm: IA} below, and the Big$-\mathcal{O}$ includes the constant $\Lambda.$ The third step is based on \eqref{2h2h} and the observation that $\ci_{2\h}^{\h}$  preserves the $L$-1 norm. As a result, $\delta_{2\h}=1-\tau_{2\h} \geq 1 - \tau_{\h} + \CO{\h}=\delta_{\h}+ \CO{\h}.$ This shows that in general the specral gap does not change significantly with $\h$.

\emph{The PDF Overlaps. --- } 
In our multilevel approach, the stationary density $p_{2\h}$ will be interpolated to approximate the stationary density $p_{\h}$, i.e., $p_{\h} \approx \ci_{2\h}^\h p_{2\h}$. We now show that these PDFs have significant overlap.

\begin{lemma}\label{lmm: IA}
Under the Assumption on the bounded variation, for each $p_\h \in \Delta_\h,$ the corresponding discrete density $\bm \pi_\h$ follows the bound,
\begin{equation}\label{iah}
    \norm{\ci_{2\h}^{\h}  \ca_{\h}^{2\h}  \bm \pi_{\h} - \bm \pi_\h}_1 \leq \Lambda  \h.
\end{equation}
Let $\bm \pi_\h$ and $\bm \pi_{2\h}$ be respectively the stationary density of $P_\h$ and $P_{2\h}$.  Then,
  \begin{equation}
      \norm{\bm \pi_\h - \ci_{2\h}^\h \bm \pi_{2\h} }_1  \leq C \frac{\h}{\delta_\h}. 
  \end{equation}
  Finally, the last inequality implies that the corresponding quantum states \eqref{quan-st} have an overlap with the following lower bound,
  \begin{equation}
  \braket{\pi_\h}{\ci_{2\h}^\h \pi_{2\h}} =
     1- q_\h, \quad \textrm{with}\;  q_\h  \leq C\frac{\h}{\delta_\h}.  
  \end{equation}
\end{lemma}

\medskip

Now we turn to the implementation of the multi-level algorithm outlined in the {\bf Diagram} \ref{diag}. For instance, for $d=1$, $\ket{\pi_{2\h}}$ has components with labels between $-1/\h$ and $1/\h$, which requires 1 less qubit to store than $\ket{\pi_{\h}}$. The interpolation operator $\ci_{2\h}^\h$ can be defined as,
\begin{equation}
\braket{x}{\pi_\h} =
  \left\{
    \begin{aligned}
            &\frac{1}{\sqrt{2}} \braket{x}{\pi_{2\h}},  \quad &\text{if} \; -\frac{1}{\h} < x< \frac{1}{\h}, \qquad \\
            &\frac{1}{\sqrt{2}}\braket{x \mp \frac{1}{\h}}{\pi_{2\h}},  \quad & \text{if} \;  \frac{1}{\h} \leq \pm x < \frac{2}{\h}.
    \end{aligned}\right.
\end{equation}
The extension to high dimensions is straightforward.

To assess the complexity of implementing each Markov chain chain, $P_\h$, we consider the quantum walk approach in  \cite[Theorem 2, $r=1$]{wocjan2008speedup}:
\begin{lemma}
For the Markov chain $P_\h$, assume that the initial density $\ket{\pi^0_\h}$ has an overlap at least $1-q$ with the stationary density $\ket{\pi_\h}$: $\braket{\pi^0}{\pi}^2\geq 1-q$, and the Markov chain has spectral gap $\delta_\h$, then after $n$ steps of the quantum walks, $W(P_\h)$, the algorithm produces an approximate density $\ket{\pi^n_\h}$ with error within $\epsilon$, 
provided that,
 \[
 n= \frac{\log \frac{1}{\epsilon}}{\sqrt{\delta} \log \frac{1}{q} } \log \left(\frac{\log \frac{1}{\epsilon}}{\sqrt{\delta} \log \frac{1}{q} }\right). 
 \]
 
\end{lemma}

Meanwhile, we notice that our algorithm involves Markov chains with varying state spaces.
Therefore, the dependence of the complexity on the state space dimension should be taken into account. Such dependence has been quantified by Chiang et al \cite{chiang2009efficient}.
\begin{lemma}\cite[Theorem 1]{chiang2009efficient}\label{lmm:chi}
 Suppose that the transition matrix on a state space with dimension $2^m$ is $s-$ sparse, and the transition matrix can be accessed with correct $t-$bit digits: $t\geq \log \frac{1}{\epsilon} + \log s$. There is an quantum algorithm that simulates each step of the Markov chain with precision $\epsilon$ and with complexity that scales linearly with $m$ and $d$, but logarithmically with $\frac{1}{\epsilon}.$
\end{lemma}

\emph{Overall Computational Complexity. ---} 
Recall that in order to prepare an initial state of $P_\h$ with state space $S_\h$, we run the Markov chain $P_{2\h}$, which according to \eqref{Sh}, has a state space with much smaller dimension $\abs{S_{2\h}} =2^{-d}\abs{S_\h}$. Thus, based on the above complexity estimate, running $P_{2\h}$ requires much less resources. The same pattern holds for $P_{4\h}, P_{8\h}, \cdots, P_{\h_\Max}.$ We can quantify the overall complexity as follows,

\begin{theorem}
Under the previous assumptions,   the multi-level approach (in {\bf Diagram} \ref{diag}) can be implemented via quantum walks with complexity that is equivalent to 
    \begin{equation}
        \mathcal{O} \left( {\frac{d\gamma }{d-1} \frac{1}{\sqrt{\delta_{\h_\Min}} 
        }
        } 
        \right),
    \end{equation}
    steps of  quantum walks 
    of $W\big(P_{\h_{\Min}}\big),$ excluding logarithmic factors.
\end{theorem}

By keeping the dominate terms, the overall complexity from Lemma \eqref{lmm:chi} is given by,
\begin{equation}
    C_\h= \frac{m_\h s_\h}{ \sqrt{\delta_\h } \log \frac{1}{q_\h} }.  
\end{equation}

Due to the fact that $s_\h \geq s_{2\h}$, and $m_\h= \log \left(\frac{2}{\h}\right)^d = \mathcal{O}\left(d \log \frac{1}{\h}\right),  $  we obtain,
\begin{equation}
   \frac{ C_{2\h}} {C_\h}   = \mathcal{O}\left(\frac{\log \frac{1}{\h} }{\log \frac{1}{q_\h} }  \frac{1}{d} \sqrt{\frac{\delta_\h}{\delta_{2\h}}}\right)  =\mathcal{O}\left( \frac{1}{d} \sqrt{\frac{\delta_\h}{\delta_{2\h}}}  \right)=\mathcal{O}\left( \frac{\sqrt{\gamma}}{d}\right). 
\end{equation}
Here we have used the bound in Lemma \ref{lmm: IA} and the assumption on the spectral gap.

Therefore, the total cost is given by,
\begin{equation}
    C_\text{total}= C_{\h_{\Min}} \big(1 + d^{-1} + d^{-2} + \cdots + d^{-L} \big) \leq \frac{d\sqrt{\gamma} }{d-1} C_{\h_{\Min}}.
  \end{equation}

The remarkable observation is that for high dimensional problems, in using low-resolution Markov chains to prepare the Markov chain $P_{\h_{\Min}},$  the complexity associated with implementing the low-resolution Markov chains is almost negligible. 

\emph{Summary. } We presented a multi-level approach to mix a Markov chain with quantum speedup.  The strategy is to create a transition to the Markov chain from low-resolution, coarse-grained Markov chains.  This effectively introduces $r$ Markov chains that can be easily initialized. This fits the general framework  of using 
 a slowly-varying sequence of Markov chains \cite{aharonov2003adiabatic,wocjan2008speedup}. But  
 we show, by leveraging the multi-level properties, that overall, the complexity, excluding logarithmic factors, is independent of the chain length.  The problem has been placed in the context of approximating a Markov chain with continuous state space. But the techniques are applicable to general Markov chains that can be reduced to a Markov chains with very small state space, through a multi-level procedure. One important class of  examples are those from multigrid  methods \cite{de2010smoothed}, including large-scale graphs. The main ingredient needed are the interpolation properties of the eigenvectors of $P_\h$. Overall, this approach adds another piece to the entire puzzle associated with the generic quadratic speed of quantum algorithms for Markov chains.  

\noindent\emph{ Acknowledgement.} The author's research is supported by the National Science Foundation Grants DMS-2111221 and a seed grant from the Institute of Computational and Data Science (ICDS) at Penn State. The author would also like to thank Dr. Patrick Rall for fruitful discussions on quantum walks.

\bibliographystyle{plain}
\bibliography{qw}

\newpage 
\appendix

\section{Interpolation error in Lemma 1}

To quickly illustrate the ideas, we compare $p_h(x)$ to $p(x)$ in $\D_{\bm j}(\h)$,
\begin{align}
      p_\h(x) -p(x) &= \h^{-d} \int_{\D_{\bm j}(\h) } p(y) dy -p(x), \\
      & = \h^{-d} \int_{\D_{\bm j}(\h) } p(y) - p(x) dy. \\
      \Longrightarrow \abs{p_\h(x) -p(x)} & \leq   \h^{-d} \int_{\D_{\bm j}(\h) }  \abs{p(y) - p(x) }dy \\
      & \leq  \Lambda \h \\
       \Longrightarrow \norm{p_\h(x) -p(x)}_1 & \leq   \Lambda \h.
\end{align}

\section{Properties of the Transfer Operators in Lemma 2}

For the second inequality in \eqref{aiia}, let $\bm  u_\h=\ci_{2\h}^{\h} \bm v_{2\h}$. 
To compute the inner product with $\bm v_\h$, we examine the summation in a bin,
\begin{equation*}
    \sum_{\bm j\in N(\h)}  \bm v_\h(\bm j)   \bm u_\h(\bm j)
    = \sum_{\bm k  \in N(2\h)} \sum_{\bm j \in \D_{k}(2\h) }   \bm v_\h(\bm j)   \bm u_\h(\bm j).
\end{equation*}
Notice that $ \bm u_\h$ is a constant in $\D_{k}(2\h)$. Therefore the last step yields the sum of $\bm v_\h$ in $\D_{k}(2\h)$, which yields $\ca_{\h}^{2\h}\bm v_\h.$

 We can express the Ulam approximation of the operator $P_\h$ in \eqref{Ph} in the following compact form,
\[ P_\h = \ca^\h K \left(\ca^\h\right)^\dagger.\]
The operator $\ca^\h$ corresponds to integrating $x$ over each subdomain $\Omega_{\bm j}(\h).$ 

Similarly,
\[ P_{2\h} = \ca^{2\h} K \left(\ca^{2\h}\right)^\dagger.\]
The identify \eqref{2h2h} follows from the observation that,
\[ \ca^{2\h}  =  \ca_{\h}^{2\h} \ca^{\h}. \]

\section{Overlap of the stationary densities in Lemma 3}

The proof of the first inequality is similar to the proof of Lemma 1.

For the second part, let $\widehat{\bm \pi}_{\h}=\ci_{2\h}^{\h} \bm \pi_{2\h}.$
Next, we check the residual error 
\begin{align*}
    P_\h \widehat{\bm \pi}_{\h} - \widehat{\bm \pi}_{\h} =& \ci_{2\h}^{\h} \ca_{\h}^{2\h}  P_\h \ci_{2\h}^{\h} \bm \pi_{2\h} - \ci_{2\h}^{\h} \bm \pi_{2\h} + \CO{\h},\\ 
    =&  \ci_{2\h}^{\h}  \left(  P_{2\h} \bm \pi_{2\h} - \bm \pi_{2\h}   \right) + \CO{\h}=\CO{\h}.
\end{align*}

Here we can let $\bm f_\h = P_\h \ci_{2\h}^{\h} \bm \pi_{2\h}$ and use the  observation that in Lemma 1.
\[
\norm{\ci_{2\h}^{\h} \ca_{\h}^{2\h}  \bm f_\h - \bm  f_\h}_1 \leq \Lambda \h.
\]

Meanwhile, by using the relation $(I - P_h) \bm \pi_\h=0,$ we gave
\begin{align}\label{eq: overproof}
   \Lambda \h \geq \norm{  P_\h \widehat{\bm \pi}_{\h} - \widehat{\bm \pi}_{\h}  }_1
       = \norm{(I - P_\h)  (\widehat{\bm \pi}_{\h} - \bm \pi_\h)}_1 \Longrightarrow &
       \norm{\widehat{\bm \pi}_{\h} - \bm \pi_\h}_1 \leq C  \frac{\h}{\delta}. 
\end{align}

In the last step, we used the fact that the spectral radius of $(I-P)^{-}$ is less that $\frac{1}{1-\tau(P_\h)}$. In addition, due to the reversibility of $P_\h$, the stochastic matrix is diagonalizable, $P_\h= V D V^{-1}$. Therefore, the inequality holds with constant $C= \norm{V}_1 \norm{V^{-1}}_1. $ Therefore, this can be considered as a Bauer–Fike type of bound.

This bound can be linked to the corresponding quantum states \eqref{quan-st}:
\begin{align*}
    \braket{\pi_\h}{ \widehat{\pi}_{\h}} = 1 - \frac{1}{2}
    \big\|\ket{\pi_\h} - \ket{\widehat{\pi}_{\h}}\big\|_2^2= 1- \frac12 \norm{\bm \pi_\h -  \widehat{\bm \pi}_\h}_1,
\end{align*}
which leads to the last statement in the Lemma. 

\medskip 

An alternative approach to derive \eqref{eq: overproof} is to use the bound by Seneta  \cite{seneta1988perturbation}. 
\begin{equation}\label{seneta}
    \norm{\widehat{\bm \pi}_{\h} - \bm \pi_\h}_1 \leq \frac{1}{\delta} \norm{P_\h - \widehat{P_{\h}}}_1.
\end{equation}
Here $\widehat{P_{\h}}=  \ci_{2\h}^\h P_{2\h}\ca_{\h}^{2\h}$ can be viewed as an interpolation of $P_{2\h}$ back to the level $\h.$ To elaborate this alternative, we first look at the property of $\widehat{P_{\h}}$ and consider it as an perturbation of $P_\h.$

\begin{proposition}
$\widehat{P_{\h}}$ is a stochastic matrix, and  the Markov chain driven by $P_{2\h}$ is equivalent to $\widehat{P}_\h$ applied to states $\ci_{2\h}^{\h} \bm \pi_{2\h}$. 
\end{proposition}

\begin{proof}
For each partition, we define the collection of piecewise constant function as $\Delta_\h.$ We notice that for any non-negative function $f\in \Delta_\h$, $\ca_{\h}^{2\h} f\geq 0$. Similarly, $\ci_{2\h}^{\h} f\geq 0$ for each $f\in \Delta_{2\h}$ and $f\geq 0.$ This implies that $\int \widehat{P}_h(x,y) f(x) g(x) dx dy \geq 0 $ for all $f, g \geq 0$ in $\Delta_{\h}.$ Furthermore, for any constant 
function $f \in \Delta_\h$, $\ci_{2\h}^{\h} f$  is the same constant function, but in $\Delta_{2\h}$. The same holds true for the operator $\ca_\h^{2\h}$.
 This shows that $\widehat{P}_h=\ci_{2\h}^{\h} P_{2h} \ca_\h^{2\h}$  has row sum 1. Combining these observations, we see that $\widehat{P}_h$ is a stochastic matrix.  
 \end{proof}
 
  $\widehat{P}_\h$ can be viewed as the perturbation of $P_\h.$ In light of \eqref{2h2h}, we need to estimate,
\begin{align}
     \norm{ P_\h -  \ci_{2\h}^\h \ca_{\h}^{2\h} P_{\h} \ci_{2\h}^\h\ca_{\h}^{2\h} }_1 \leq &
 \norm{P_\h  - \ci_{2\h}^\h \ca_{\h}^{2\h} P_{\h}}_1 
 + \norm{  \ci_{2\h}^\h \ca_{\h}^{2\h} ( P_{\h} -  P_{\h} \ci_{2\h}^\h\ca_{\h}^{2\h}) }_1,\\
 \leq &
 \norm{P_\h  - \ci_{2\h}^\h \ca_{\h}^{2\h} P_{\h}}_1 
 + \norm{  P_{\h} -  P_{\h} \ci_{2\h}^\h\ca_{\h}^{2\h} }_1.
\end{align}
Using an approximation property similar to \eqref{iah}, one can show that the right hand side is $\mathcal{O}(\h).$ So together with \eqref{seneta}, one arrives at a similar bound in \eqref{eq: overproof}.

\end{document}